\documentclass[11pt]{article}
\usepackage{chao}
\usepackage{hyperref}
\usepackage{cleveref}
\usepackage{cite}
\usepackage{verbatim}
\usepackage[T1]{fontenc}

\def\abs#1{\left| #1 \right|}

\DeclareMathOperator{\MST}{MST}

\begin{document}

\title{An Optimal Algorithm for the Stacker Crane Problem on Fixed Topologies\footnote{Supported by Science and Technology department of Sichuan Province under grant M112024ZYD0170.}} 
\author{Yike Chen\footnote{ \email{cyike9982@gmail.com}, University of Electronic Science and Technology of China.} \and Ke Shi\footnote{ \email{self.ke.shi@gmail.com}, University of Science and Technology of China.} \and Chao Xu\footnote{ \email{the.chao.xu@gmail.com}, University of Electronic Science and Technology of China.}
}
\date{}
\maketitle

\begin{abstract}
The Stacker Crane Problem (SCP) is a variant of the Traveling Salesman Problem. In SCP, pairs of pickup and delivery points are designated on a graph, and a crane must visit these points to move objects from each pickup location to its respective delivery point. The goal is to minimize the total distance traveled. SCP is known to be NP-hard, even on trees. The only positive results, in terms of polynomial-time solvability, apply to graphs that are topologically equivalent to a path or a cycle.

We propose an algorithm that is optimal for each fixed topology, running in near-linear time. This is achieved by demonstrating that the problem is fixed-parameter tractable (FPT) when parameterized by both the cycle rank and the number of branch vertices.
\end{abstract}

\section{Introduction}

In the stacker crane problem (SCP), a stacker crane must retrieve and deliver a set of \(m\) items from and to specified locations within a warehouse. The goal is to determine an optimal order for these operations, constructing a tour that completes the task while minimizing the total distance traveled. The problem is typically modeled as a problem on a weighted graph, where the stacker crane moves between vertices, traversing the edges.

The SCP was first studied by Frederickson et al. \cite{FredericksonHK78}, who provided an equivalent formulation involving a tour on a mixed graph and designed a \(9/5\)-approximation algorithm. Subsequently, there have been many studies focused on developing approximation algorithms for various SCP variants \cite{10.1007/978-3-030-92681-6_32,10.1007/978-3-319-71150-8_8,yu_approximation_2023}.

Regarding exact algorithms, Frederickson and Guan demonstrated that SCP is NP-hard even on trees \cite{FredericksonG93}, ruling out some of the simplest graph classes as candidates for polynomial-time solvability. However, Atallah and Kosaraju showed that when the input graph is a path, SCP can be reduced to a minimum spanning tree problem. Moreover, if the input graph is a cycle with \(n\) vertices, there exists an \(O(m+n\log n)\)-time exact algorithm \cite{AtallahK88}. Frederickson later improved the running time to match that of finding a minimum spanning tree, and also proved the algorithm to be optimal \cite{Frederickson93}.

Since paths and cycles are topologically equivalent to a single edge or a self-loop, respectively, one might conjecture that SCP is solvable in polynomial time for any fixed topology. Fixed topologies have practical significance: the layout of warehouses is usually fixed, often resembling a grid with a limited number of aisles \cite{ROODBERGEN200132,CAMBAZARD2018419}. For further details on various industrial applications of SCP, we refer readers to \cite{doi:10.1137/1.9781611973679}.

\paragraph*{Our Contributions}

We present an optimal algorithm for solving the SCP for fixed topologies. The idea is to design an FPT algorithm parameterized by the cycle rank and the number of branch vertices. Our algorithm generalizes Frederickson’s approach, reducing to his algorithm when the topology is a path or a cycle \cite{Frederickson93}. Because our algorithm is more general, it also serves as an alternative and more intuitive proof of the correctness of Frederickson's algorithm for paths and cycles as special cases \cite{Frederickson93}.


\section{Preliminaries}

We consider mixed graphs, which can contain both undirected edges and directed edges, the latter being referred to as \emph{arcs}. A graph that contains only arcs is called a directed graph, while one that contains only undirected edges is called an undirected graph. A vertex is termed a \emph{branch} vertex if its degree is at least 3.

For an arc \((s,t)\), \(s\) is the tail and \(t\) is the head. For an undirected edge, both vertices are tails and heads. A \emph{walk} is a sequence of edges such that the tail of the subsequent edge matches the head of the preceding edge. A \emph{tour} is a walk that starts and ends at the same vertex. Let \(\MST(m)\) denote the worst-case running time for finding a minimum spanning tree in a graph with \(m\) edges.

We now provide a graph-theoretical definition of SCP. In SCP, we are given a simple graph \(B=(V,E)\), called the \emph{base graph}, together with a list of \emph{requests} \(R\). A request is an ordered pair of vertices, represented as an arc. For simplicity, we assume that the requests are distinct, although our results can handle cases with duplicate requests.

We use the following formulation of SCP, as introduced by Frederickson et al. \cite{FredericksonHK78}.

\begin{problem}[Stacker Crane Problem (SCP)]
  Given a simple graph \(B=(V,E)\) on $n$ vertices and a set of $p$ arcs \(R\), with a cost function \(c:E\cup R \to \R^+\), find a min-cost tour in the mixed graph \(G=(V,E\cup R)\) that traverses each arc in \(R\) exactly once.
\end{problem}

Intuitively, by assigning the cost of an arc \((s,t) \in R\) to be the length of the shortest \(st\)-path in \(B\), traversing an arc represents moving an object from \(s\) to \(t\). We emphasize that this ensures the object must be delivered immediately once it is picked up.

An \emph{edge subdivision} operation involves inserting a new vertex in the middle of an edge. Specifically, given an edge \(uv\), we delete the edge, add a new vertex \(w\), and connect \(u\) and \(v\) via \(w\) by adding edges \(uw\) and \(wv\). A graph \(H\) is called a \emph{subdivision} of a graph \(G\) if \(H\) can be obtained through a sequence of edge subdivision operations starting from \(G\). Two graphs, \(G_1\) and \(G_2\), are \emph{topologically equivalent} if $G_1$ and $G_2$ are subdivisions of the same graph $G$. Our goal is to show that SCP for a fixed topology can be solved in optimal time.

\begin{problem}[Stacker Crane Problem with Fixed Topology (SCP\((H)\))]
  Given a simple graph \(B=(V,E)\) on $n$ vertices that is topologically equivalent to a graph \(H\), and a set of $p$ arcs \(R\), with a cost function \(c:E\cup R \to \R^+\), find a min-cost tour in the mixed graph \(G=(V,E\cup R)\) that traverses each arc in \(R\) exactly once.
\end{problem}

\subsection{Circulations}
Consider a mixed graph \( G = (V, E \cup A) \) with a set of undirected edges \( E \) and directed arcs \( A \). An arbitrary orientation is assigned to each undirected edge to define a forward direction, which serves solely as a notational convenience and has no bearing on the results. We define \( \delta^-(v) \) and \( \delta^+(v) \) as the sets of in-edges/in-arcs and out-edges/out-arcs of a vertex \( v \), respectively. Each edge and arc has an associated lower bound, upper bound, and cost function, denoted as \( \ell, u: E \cup A \to \mathbb{Z} \) and \( c: E \cup A \to \mathbb{R} \).

A \emph{circulation} is a function \( f: E \cup A \to \mathbb{Z} \) such that \(\sum_{e \in \delta^-(v)} f(e) = \sum_{e \in \delta^+(v)} f(e)\) for every vertex \( v \in V \). The circulation \( f \) is \emph{feasible} with respect to the lower bounds \( \ell \) and upper bounds \( u \) if \( \ell(e) \leq f(e) \leq u(e) \) for every edge \( e \in E \cup A \). Typically, the values of \( \ell \) and \( u \) are clear from the context. 

The cost of a circulation \( f \) under the cost function \( c \) is given by \( \sum_{a \in A} f(a)c(a) + \sum_{e \in E} |f(e)|c(e) \). A feasible circulation with the minimum cost is called a \emph{min-cost circulation}. An edge or arc \( e \) is termed \emph{fixed} if \( \ell(e) = u(e) \). An edge \( e \) is considered \emph{unconstrained} if \( \ell(e) = -\infty \) and \( u(e) = \infty \). An arc is \emph{unconstrained} if \( \ell(e) = 0 \) and \( u(e) = \infty \). The \emph{support} of \( f \), denoted by \( \supp(f) \), is defined as \( \{ e \mid f(e) \neq 0, e \in E \cup A \} \). 

For a circulation \( f \), we use \( [f] \) to denote the set of all tours corresponding to \( f \) in a natural way: a tour where \( f(e) \) represents the difference between the number of times the tour traverses the edge in the forward and backward directions. The set \( [f] \) is referred to as a \emph{homology class}, and two tours in \( [f] \) are said to be \emph{homologically equivalent}.

For the SCP problem with input graph \( G = (V, E \cup R) \), the \emph{corresponding circulation problem} is a circulation problem on \( G \) where each edge in \( E \) is unconstrained, and each arc in \( R \) is fixed with a flow value of 1. The costs of the edges and arcs remain the same as in the original SCP formulation. 
The SCP problem and the corresponding circulation problem are not equivalent. A gap exists due to connectivity. A tour must be connected, but the corresponding circulation does not guarantee connectivity by simply looking at the support.

\subsection{Cycle space}
For an undirected graph $G = (V, E)$, the functions $f: E \to \mathbb{Z}$ can be represented as vectors in $\mathbb{Z}^E$. The space of circulations is a submodule of $\mathbb{Z}^E$, known as the \emph{integral cycle space}, or simply the cycle space. If the graph is connected, the dimension of the cycle space is $|E| - |V| + 1$ \cite{Bollobas2013-ad}. This value, $|E| - |V| + 1$, is referred to as the \emph{cycle rank}. Any two topologically equivalent graphs have the same cycle rank.

A circulation $f$ is called a \emph{cycle flow} if $\supp(f)$ forms an undirected cycle (after considering all edges as undirected). A cycle flow $f$ is a \emph{unit cycle flow} if $|f(e)| = 1$ for all $e \in \supp(f)$. Let $T$ be a spanning tree in $G$. For each edge $uv \notin T$, the \emph{fundamental cycle} $C_{uv}$ with respect to $T$ is defined as the unit cycle flow that traverses $uv$ in the forward direction and then traces back from $v$ to $u$ along the unique path in $T$. Note that it is possible that some fundamental cycles have a negative flow value on some edges. 

The fundamental cycles with respect to $T$ form a basis for the cycle space. That is, let $C_1, \ldots, C_r$ be the fundamental cycles with respect to $T$. Then, for any circulation $f$, we have $f = \sum_{i=1}^r \lambda_i C_i$ for some coefficients $\lambda \in \mathbb{Z}^r$ \cite{LIEBCHEN2007337}.

\subsection{Techniques}
Atallah and Kosaraju solved SCP$(H)$ for the cases where $H$ is either a path or a cycle \cite{AtallahK88}. Using a more modern and general perspective, their approach can be summarized as a three-step algorithm:

\begin{enumerate}
    \item Find a min-cost circulation $f$ in the corresponding circulation problem.
    \item Find all circulations $g$ such that $\|f-g\|_\infty \leq b$ for some bound $b$.
    \item Find a min-cost tour in each homology class $[g]$, and return the minimum.
\end{enumerate}

It is crucial that $H$ is a path or a cycle, as their technique in each of these steps relies heavily on this property. We describe how to eliminate such a restriction.

In the first step, finding a min-cost circulation when $H$ is either a path or a cycle was solved through an ad-hoc process in \emph{linear time}. On a path, such a circulation is unique. For a cycle, they fix a flow sent along a particular edge, and then the circulation becomes unique. Alternatively, one can use a min-cost flow algorithm on a one-tree \cite{Kalantari1995}, which can also obtain a linear time algorithm. We show that min-cost circulation can be solved in linear time for arbitrary fixed $H$ by reducing it to a problem similar to linear programming in a fixed dimension. See \autoref{sec:mincostcirc}.

For the second step, Atallah and Kosaraju showed that $b = 0$ suffices for paths, leading to an optimal algorithm. In the same paper, they established a bound of $b = p$ for cycles, where $p$ is the number of requests. Fortunately, their technique does not explicitly generate all circulations, but instead searches through them using a binary search, implemented alongside step 3. Later, Frederickson improved the bound on $b$ to $1$ by squeezing the value of the min-cost tour between two one-dimensional functions, thereby achieving the desired proximity result \cite{Frederickson93}. This improvement not only made the algorithm faster but also allowed steps 2 and 3 to be separated.

In general, it is not difficult to argue that for any $H$, a bound of $b = p$ suffices. The number of feasible circulations in step 2 would then be $O((2p+1)^r)$, where $r$ is the cycle rank of $B$. This is too large to yield a fast algorithm. Unfortunately, Frederickson's technique does not generalize to graphs with higher cycle ranks. Therefore, we need a proximity theorem for arbitrary graphs. This is the most technical part of our paper. As we will show in \autoref{sec:proximityresult}, a bound of $b = r$ suffices. The algorithm to explicitly enumerate $O((2r+1)^r)$ such circulations will be described in \autoref{sec:enum}.

For the third step, Atallah and Kosaraju found the min-cost tour in a homology class by computing a minimum spanning tree in an auxiliary graph. The auxiliary graph is constructed by contracting edges in the circulation that defines the homology class. Once again, computing a minimum spanning tree suffices only because $H$ is a path or a cycle. In \autoref{sec:tourinhomo}, we show that the correct problem to solve is \emph{minimum Steiner tree} in the auxiliary graph, which is an NP-hard problem in general. However, since the number of Steiner branch vertices depends only on $H$, the running time remains efficient.

\section{Proximity result}\label{sec:proximityresult}
Let $B = (V, E)$ be the base graph. Consider the mixed graph $G = (V, E \cup R)$ in the SCP problem and the corresponding circulation problem on $G$. Assume $f$ is the min-cost circulation in $G$. We will show that there exists a circulation $g$ in $G$ such that the homology class $[g]$ contains the min-cost tour, and $\|f - g\|_\infty \leq r$, where $r$ is the cycle rank of $B$.

For a unit cycle flow $C$, we call $kC$ a cycle flow of value $k$. A circulation $f$ is said to contain a circulation $g$ if $|g(e)| \leq |f(e)|$ for each edge $e$, and $g(e)$ and $f(e)$ have the \emph{same sign}. 

\begin{definition}
  A circulation is  \emph{elementary} if it does not contain a cycle flow of value $2$. 
\end{definition}

Intuitively, this means that after cancellations, no part of the flow circulates around a cycle more than once. In fact, we prove the proximity result by showing that $f - g$ is an elementary circulation in $B$, the undirected base graph.

\subsection{Elementary Circulations}

For a graph or mixed graph $G$, let $G' = \text{sym}(G)$ be the directed graph obtained by replacing each edge with two opposing arcs. Interestingly, we prove the proximity result by considering circulations in the residual graph of $G'$, as it is easier to reason about.

The \emph{residual graph} for a given directed graph $G$ with lower and upper bounds $\ell$ and $u$, and a feasible circulation $f$, is the directed graph $G_f$ along with new lower and upper bounds $\ell_f$ and $u_f$. For each arc $a \in A$, define $u_f(a) = u(a) - f(a)$ and $\ell_f(a) = \ell(a) - f(a)$. The residual graph $G_f$ consists of all arcs for which either $u_f(a) \neq 0$ or $\ell_f(a) \neq 0$. Note that the lower bounds in the residual graph do not have to be non-negative. However, $\ell_f(a) \leq 0 \leq u_f(a)$ is always satisfied, meaning that the zero circulation is always feasible. The fixed arcs in the original graph do not appear in the residual graph, so the complexity of the residual graph can be much smaller than the original graph.

One can observe that the SCP can be reduced to a min-cost \emph{connected} circulation problem on $G' = (V, A \cup R) = \text{sym}(G)$. Here, $A$ are the arcs obtained from symmetrizing $E$. For arcs in $A$, the lower bound is $0$ and the upper bound is infinity. The cost of each arc is the cost of the corresponding edge. For arcs in $R$, both the lower bound and upper bound are $1$. If $f'$ is a min-cost connected circulation in $G'$, then it corresponds to a min-cost tour of the same cost in $G$.

\begin{theorem}\label{thm:diffminconn}
  Let $f$ be a min-cost circulation in $G = (V, E \cup R)$, where the flows on arcs in $R$ are fixed to $1$, and edges $E$ are unconstrained. There exists a circulation $g$ such that the min-cost tour is in the homology class $[g]$ and $g - f$ is an elementary circulation in the base graph $B = (V, E)$.
  \end{theorem}

  \begin{proof}
  Let $\Delta = g - f$. We choose $g$ such that a min-cost tour is contained in $[g]$, and out of all such $g$, one that minimizes $\|\Delta\|_1$. Note that $\Delta$ is a circulation in $B$ because $f$ and $g$ agree on the fixed arcs in $R$, hence $\Delta$ is zero outside of $E$.
  
  Let $G' = \text{sym}(G)$, and let $f'$ be the min-cost circulation in $G'$, which corresponds to the min-cost circulation $f$ in $G$ in a natural way. 
  Specifically, for each edge $uv$, let $(u,v)$ denote the positive orientation.
  Define $f'(u,v) = f(uv)$ and $f'(v,u) = 0$ if $f(uv) \geq 0$
 and otherwise set $f'(u,v) = 0$ and $f'(v,u) = -f(uv)$. 
 This ensures that $f(uv) = f'(u,v) - f'(v,u)$.
  Let $g'$ be the min-cost connected circulation in $G'$, such that $g(uv) = g'(u,v) - g'(v,u)$, again using the positive orientation $(u,v)$. Let $\Delta' = g' - f'$, which is a circulation in $B' = \text{sym}(B)$.
  
  Consider the residual graph $G'_{f'}$. Certainly, $\Delta'$ is a feasible circulation in $G'_{f'}$ because $f' + \Delta' = g'$ is a feasible circulation in $G'$.
  We label the arcs in $G'_{f'}$ as follows: for each edge $e \in E$,
  label two corresponding arcs as $e^+$ and $e^-$ such that 
  $\Delta'(e^+) - \Delta'(e^-) \geq 0$, which
  implies $\Delta'(e^+) - \Delta'(e^-) = \abs{g'(e^+)-f(e^+)-g'(e^-)+f(e^-)} = \abs{g(e)-f(e)}  =\abs{\Delta(e)}$.
  
%

  Assume $\Delta$ is not elementary. Then, there exists a unit cycle flow $C$ such that $2C$ is contained in $\Delta$. In terms of $G'_{f'}$, this would imply that $\Delta'(e^+) - \Delta'(e^-) \geq 2$ for each $e \in C$. We consider a cycle flow $C'$ in $G'_{f'}$ whose undirected version is $C$, and which uses the maximum number of edges with positive labels. That is, for each edge $e$, we set $C'(e^+) = 1$ and $C'(e^-)=0$ if $\Delta'(e^+) \geq 2$, 
  otherwise, since $\Delta'(e^+) \leq 1$, it follows that $\Delta'(e^-) \leq \Delta'(e^+)-2 \leq -1$,   
  set $C'(e^+) = 0$ and $C'(e^-)= -1$.
  
  Certainly, $\Delta' - C'$ is feasible in $G'_{f'}$. The cost of $C'$ cannot be negative, since the residual graph of a min-cost circulation cannot contain a negative cost cycle \cite{AhujaMO93}. This implies that the cost of $g' - C'$ is no greater than the cost of $g'$.
  
  Additionally, we need to show that $\supp(g' - C')$ is connected. For each $e \in \supp(C)$, there are two cases:
  
  \begin{enumerate}
      \item If $C'(e^+) = 1$, then $\Delta'(e^+) \geq 2$. Therefore, $g'(e^+) - C'(e^+) = f'(e^+) + \Delta'(e^+) - 1 \geq 1$. Hence, $e^+ \in \supp(g' - C')$.
      \item If $C'(e^-) = -1$, then $\Delta'(e^-) \leq -1$. Therefore, $g'(e^-) - C'(e^-) = g'(e^-) + 1 \geq 1$. Hence, $e^- \in \supp(g' - C')$.
  \end{enumerate}
  
  Since $\supp(g' - C')$ would include at least one of $e^+$ or $e^-$ for each $e \in C$, if $\supp(g')$ is connected, then $\supp(g' - C')$ is also connected. Because $g' - C'$ is a connected circulation of no greater cost than $g'$, we have that $[g - C]$ also contains a min-cost tour. However, $\|\Delta - C\|_1 < \|\Delta\|_1$ because $\Delta$ contains $C$. This contradicts our choice of $\Delta$ having minimal $L_1$ norm.
  \end{proof}

  \subsection{Elementary Circulations and $L_\infty$ Norm}

  In this section, we bound the number of elementary circulations in a \emph{graph} using its cycle rank.

  The idea is to show that for an elementary circulation $f$, we have $\|f\|_\infty \leq r$, where $r$ is the cycle rank, by expressing the circulation as the sum of $r$ cycle flows. We then use the fact that no flow traverses a cycle more than once to establish the desired bound.

  However, a naive decomposition might result in a cycle flow with values greater than $2$, even though the overall circulation remains elementary due to cancellations.

  To avoid this issue, we transition to non-negative circulations, ensuring that everything we analyze is non-negative. This simplifies the argument, as it prevents flow cancellations on edges when summing flows around cycles.

  First, we strengthen the well-known flow decomposition theorem for circulations, which typically states that a circulation can be decomposed into $m$ cycle flows, where $m$ is the number of edges \cite{Williamson_2019}. We show that this can be improved to $r$, the cycle rank of the graph.

  \begin{proposition}\label{thm:flowdecomposition}
  Every non-negative circulation on a graph with cycle rank $r$ can be decomposed into at most $r$ non-negative cycle flows. That is, $f = \sum_{i=1}^r \lambda_i C_i$, where $C_1, \ldots, C_r$ are non-negative cycle flows and $\lambda_1, \ldots, \lambda_r \geq 0$.
  \end{proposition}
  
  \begin{proof}
  We prove this by induction on the cycle rank $r$. 
  
  If the cycle rank is $0$, then the graph contains no cycles, and $f$ must be zero everywhere. 
  
  Now, consider the case where the cycle rank is $r$. We can assume the graph is connected; otherwise, we can apply the proof to each connected component separately. Suppose there exists a cycle $C$ with positive flow. Let $\lambda = \min_{e \in C} f(e)$. We include $\lambda C$ as a term in the decomposition.
  
  Let $D$ be the set of edges with zero flow after reducing the flow on $C$ by $\lambda$. Assume that $G - D$ has $k$ components. Observe that $k \leq |D| + 1$. If $k = |D| + 1$, this would imply that removing each edge in $D$ disconnects the graph, contradicting the fact that $G$ contains the cycle $C$. Hence, $k \leq |D|$.
  
  Let the components of $G - D$ be $V_1, \ldots, V_k$. By the inductive hypothesis, the remaining circulation can be decomposed into $\sum_{i=1}^k \left( |E_i| - |V_i| + 1 \right) = (m - |D|) - n + k \leq m - n = r - 1$ non-negative cycle flows. Therefore, we can decompose $f$ into $r$ non-negative cycle flows.
  \end{proof}

  Next, we show that elementary circulations have a small $L_\infty$ norm.
  
  \begin{proposition}\label{thm:elementaryproximity}
  Let $f$ be an elementary circulation in a graph with cycle rank $r$. Then, $\|f\|_\infty \leq r$.
  \end{proposition}
  
  \begin{proof}
  Let $f$ be an elementary circulation. Assume $f$ is non-negative; if not, we reverse the orientation of the edges with negative flow so that $f$ becomes non-negative. By \Cref{thm:flowdecomposition}, $f$ can be decomposed into $r$ non-negative cycle flows. Since $f$ is an elementary circulation, each of the $r$ cycle flows is a unit cycle flow. Hence, for each edge $e$, we have $f(e) = \sum_{i=1}^r C_i(e) \leq \sum_{i=1}^r 1 = r$. Therefore, $\|f\|_\infty \leq r$.
  \end{proof}

The bound in \autoref{thm:elementaryproximity} is tight. Consider a cycle rank $r$ graph with $2(r+1)$ vertices that consists of a path $P=v_1,\ldots,v_{2r}$. There are also edges $e_{r-i}$ between $v_{2r-i}$ and $v_{1+i}$ for $i\leq r$. See \autoref{fig:tightgraph}. It has cycle rank $r$. Let $C_i$ be the fundamental cycle with respect to the path $P$ using edge $e_i$ in the smaller to larger vertex direction, where $i\geq 1$. Consider the circulation $f=\sum_{i=1}^r C_i$, which is elementary, and $f(e_0)=r$.

\begin{figure}[ht]
  \centering
  \includegraphics{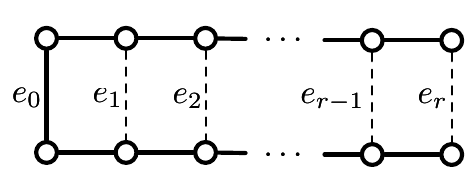}
  \caption{Example graph where there is an elementary circulation where $\|f\|_\infty=r$.}
  \label{fig:tightgraph}
\end{figure}

\section{The new algorithm}\label{sec:algo}

This section describes the optimal algorithm up to constants for SCP$(H)$.

\subsection{Min-Cost Circulation}\label{sec:mincostcirc}

Recall that we are interested in finding the min-cost circulation $f$ of $G = (V, E \cup R)$, where $f(a) = 1$ for all $a \in R$. The min-cost circulation problem was shown to be solvable in $m^{1+o(1)}$ time, where $m$ is the number of edges in the graph \cite{9996881}. In our case, the graph has $p + m$ edges. The min-cost flow algorithm on a one-tree gives a linear time algorithm for the special case when the base graph have cycle rank $1$ \cite{Kalantari1995}. 
We can show that there is a linear time algorithm for base graphs with fixed cycle rank. 

\begin{theorem}\label{thm:mincostcycle}
Let $r$ be a fixed number. Consider a min-cost circulation problem where the edges are either fixed or unconstrained, and the unconstrained edges form a graph with cycle rank $r$. The min-cost circulation can be found in $O(m)$ time.
\end{theorem}

\begin{proof}
Let $B = (V, E)$ be the undirected graph containing all the unconstrained edges, and let $R$ be the set of fixed edges. We have $G = (V, E \cup R)$. Consider a spanning tree $T$ on $B$, and let the edges in $E \setminus T$ be $e_1, \ldots, e_r$. We define $g(\lambda_1, \ldots, \lambda_r)$ to be the cost of the min-cost circulation $f$ in $G$ such that $f(e) = 1$ for $e \in R$ and $f(e_i) = \lambda_i$ for $1 \leq i \leq r$. To find the min-cost circulation, we need to solve $\min_\lambda g(\lambda)$. Note that there is an integral minimizer, as $\lambda_i$ represents the flow value on edge $e_i$, which can always be taken to be integral.

Finding the minimum is equivalent to solving the following linear program. For each $a \in R$, compute the fundamental cycle $C_a$, and let $b_e = \sum_{e \in C_a} C_a(e)$. For each $e \in E$ and $1 \leq i \leq r$, we define $A_{e, i} = C_{e_i}(e)$. 

Here, $A$ is the edge-fundamental cycle incidence matrix, which is a network matrix \cite{LIEBCHEN2007337}.

\begin{equation*}
  \begin{array}{ll@{}ll}
  \text{min}_{\lambda \in \mathbb{R}^{E \setminus T}, x \in \mathbb{R}^{T}} & \displaystyle\sum\limits_{e} x_e & \\
  \text{subject to} & x_e \geq c(e) \left| A_e \cdot \lambda + b_e \right| & \quad \forall e \in E \\
  \end{array}
\end{equation*}

Zemel showed that the minimizer of mathematical programs of the above form can be seen as a generalization of $r$-dimensional $L_1$ linear regression, which can be solved using techniques similar to Megiddo's constant-dimension linear programming algorithm \cite{Megiddo84}. Consequently, the problem can be solved in $2^{O(2^r)} m = O(m)$ time \cite{zemel_on_1984}.
\end{proof}

\subsection{Generating Circulations Near the Min-Cost Circulation}\label{sec:enum}

In this section, we show how to generate all circulations that are close to the min-cost circulation in terms of the $L_\infty$ norm.

\begin{lemma}\label{lem:linftyimpliescounting}
For an undirected graph with cycle rank $r$, the number of circulations $f$ such that $\|f\|_\infty \leq k$ is at most $(2k + 1)^r$ and these circulations can be found in $O((2k + 1)^r m)$ time.
\end{lemma}

\begin{proof}
Consider any fundamental cycle basis of $G$, a graph with cycle rank $r$, with respect to some spanning tree $T$. The values on the edges outside $T$ uniquely determine the circulation. For a circulation $f$ such that $\|f\|_\infty \leq k$, the absolute value of the flow on each edge outside $T$ is at most $k$. Hence, there can be at most $2k + 1$ choices for each of the $r$ edges outside $T$. This shows that there can be at most $(2k + 1)^r$ circulations with $L_\infty$ norm at most $k$. 

To construct all such circulations, we find a fundamental cycle basis $C_1, \ldots, C_r$. This takes $O(rm)$ time. Next, we generate circulations of the form $\sum_{i=1}^r \lambda_i C_i$ one by one where $-k \leq \lambda_i \leq k$. The sequence of circulations is generated using a generalized Gray code, so two consecutive circulations differ in only a single cycle \cite{Guan1998GENERALIZEDGC}. Therefore, it takes $O(m)$ time to generate a circulation from the previous circulation by augmenting a cycle. The total running time for generating all circulations with $L_\infty$ norm no larger than $k$ is $O((2k + 1)^r m)$.
\end{proof}

\begin{corollary}\label{cor:closecirculations}
Let $f$ be the min-cost circulation in $G$. The number of circulations $g$ in $G$ such that $\|f - g\|_\infty \leq r$ is $O((2r + 1)^r)$ and they can be found in $O((2r + 1)^r m)$ time, where $m$ is the number of edges in $G$.
\end{corollary}

\begin{proof}
By \autoref{lem:linftyimpliescounting}, there can be at most $(2r + 1)^r$ circulations with $L_\infty$ norm no larger than $r$ in $B$. These circulations can be computed in $O((2r + 1)^r m)$ time. For each such circulation $h$, we compute $g = f + h$. This takes an additional $O(m)$ time per circulation. Thus, the total time complexity is $O((2r + 1)^r m)$.
\end{proof}

\subsection{Min-Cost Tour in a Given Homology Class}\label{sec:tourinhomo}

Recall that we are interested in finding the min-cost tour in the homology class $[f]$ for some given circulation $f$ in $G = (V, E \cup R)$, where $f(a) = 1$ for $a \in R$. 

If $f$ is already connected, then we are done, as we can find a tour that uses each edge $e$ exactly $|f(e)|$ times. Otherwise, we need to find a tour in $[f]$. Note that in addition to the edges in $\supp(f)$, some extra edges might need to be used by the tour. Naturally, this would look like a Steiner Asymmetric Traveling Salesman Problem (ATSP), defined as finding the shortest tour in a directed graph that visits each terminal vertex at least once, which is even harder to work with than ATSP. However, because we are restricting the tour to be in $[f]$, this limits the kinds of edges we add to the tour: Specifically, these are edges outside $\supp(f)$ that must be traversed both forward and backward \emph{exactly once}. 

To address this, we take $G$ and contract each connected component of $\supp(f)$ into a single vertex. Let the resulting graph be $B'$. Note that only edges in $E$ might remain, as all arcs in $R$ have been contracted. Therefore, $B'$ is a graph instead of a mixed graph. Moreover, we change the weight of the edges to twice their original cost since each edge would be traversed once forward and once backward. We then reduce the problem to the minimum Steiner tree problem. 

For a graph $G = (V, E)$ with a set of terminal vertices $T \subseteq V$, the vertices in $V \setminus T$ are called \emph{Steiner vertices}. A Steiner tree is a tree that contains all the vertices in $T$. The minimum Steiner tree with respect to a weight function $w: E \to \mathbb{R}^+$ is a Steiner tree of minimum weight.

We are interested in solving the minimum Steiner tree problem on $B'$, where the terminals are the contracted vertices. This would give us the minimum set of extra edges the tour must traverse to be connected. The minimum Steiner tree problem is NP-hard in general, but note that $B'$ has only a bounded number of Steiner branch vertices. The number of Steiner branch vertices in $B'$ is at most the number of branch vertices in $B$. Indeed, the contracted vertices are not Steiner vertices, so all Steiner branch vertices must be branch vertices in $B$. The Steiner vertices of degree $1$ and $2$ can be preprocessed away, so only Steiner branch vertices remain \cite{Duin00}. Hence, the minimum Steiner tree can be solved in $O(2^k \MST(m))$ time for a graph containing $k$ Steiner branch vertices and $m$ edges \cite{Hakimi71}.

Since $B'$ can have at most $m$ edges, the minimum Steiner tree in $B'$ can be found in $O(2^k \MST(m))$ time. This leads us to the following theorem.

\begin{theorem}\label{thm:tourinhomo}
For an input mixed graph $G = (V, E \cup R)$, let $f$ be a circulation where $f(a) = 1$ for each $a \in R$. Finding a min-cost tour in $[f]$ takes $O(p + 2^k \MST(m))$ time, where $m = |E|$, $p = |R|$, and $k$ is the number of branch vertices in $B = (V, E)$.
\end{theorem}

\subsection{Putting Everything Together}

Combining all the components, we obtain the desired algorithm as described in \autoref{fig:algo}. We are now ready to prove our main theorem.

\begin{theorem}\label{thm:main}
SCP$(H)$ with an input graph of $n$ vertices and $p$ requests can be solved in $O(\MST(n) + p)$ time.
\end{theorem}

\begin{proof}
Let $B$ be the base graph and $R$ be the set of requests. For a fixed topology, the cycle rank $r$ and the number of branch vertices $k$ are constants. If $m$ is the number of edges in $B$, then $m = n + r - 1 = O(n)$.

Consider the algorithm in \autoref{fig:algo}. By \autoref{thm:mincostcycle}, computing the min-cost circulation takes $O(m + p)$ time. By \autoref{cor:closecirculations}, enumerating all $O(1)$ possible circulations close to the min-cost circulation takes $O(m)$ time. By \autoref{thm:tourinhomo}, finding the min-cost tour in the homology class of each circulation takes $O(\MST(m) + p)$ time. Using the fact that $m = O(n)$, we obtain a final running time of $O(\MST(n) + p)$.
\end{proof}

The running time in \autoref{thm:main} is optimal. The term $\MST(n)$ is unavoidable, as finding the minimum spanning tree on $n$ edges can be reduced to SCP$(P)$ in linear time, where $P$ is a path \cite{Frederickson93}. Moreover, since reading the input requires at least $O(p)$ time, the term $p$ is also tight.

Now, in terms of FPT, what we showed is that the SCP problem with the cycle rank and the number of branch vertices as parameters is FPT.

      \begin{figure*}[ht]
        \begin{center}
        \begin{algo}
          \textul{$SCP(B=(V,E), R, c)$:}\+
          \\  $G\gets (V,E\cup R)$
          \\  $f\gets $ the min-cost circulation
          \\  $\mathcal{C} \gets$ all circulations $g$ such that $\|f-g\|_\infty\leq r$ 
          \\  for $g \in \mathcal{C}$\+
          \\      record $\textsc{MinCostTourInHomologyClass}(G,g,c)$ as a candidate\-
          \\  return the min-cost candidate
          \-
        \end{algo}
        \end{center}
        \caption{Pseudocode for solving $SCP$ given base graph $B$, requests $R$ and cost function $c$.}
        \label{fig:algo}
      \end{figure*}

      \begin{figure*}[ht]
        \begin{center}
        \begin{algo}
          \textul{$\textsc{MinCostTourInHomologyClass}(G, f, c)$:}\+
          \\  \Comment{Find the min-cost tour in $[f]$}
          \\  $H\gets $ contract all components of $\supp(f)$ in $G$
          \\  for each edge $e\in V(H)$\+
          \\       $w(e)\gets 2c(e)$\-
          \\  $T\gets $ the vertices from contracting a component of $\supp(f)$
          \\  $S\gets$ Minimum Steiner tree on $H$ with weights $w$ and terminals $T$
          \\  $W\gets $ construct a tour from $S$ and $f$
          \\  return $W$
          \-
        \end{algo}
        \end{center}
        \caption{Finding the min-cost tour for a given graph $G$ and a circulation $f$ on $G$}
        \label{fig:opttour}
      \end{figure*}

      \section{Discussion}

      We simplified the presentation of our results, but our approach can be extended in multiple ways without increasing the time complexity.
      
      As mentioned earlier in the preliminaries, we can handle duplicate arcs in $R$ without altering the complexity of the algorithm. Indeed, this is done by allowing demands on arcs in $R$. Specifically, we can introduce a function $d: R \to \mathbb{N}$ such that, for each $a \in R$, the tour is required to traverse $a$ exactly $d(a)$ times. This generalization can be handled by simply fixing the value of the flow to be $d(a)$ for each $a\in R$ during our computation.
      
      Our current cost function is symmetric; for an edge $e = uv$, traversing from $u$ to $v$ or from $v$ to $u$ incurs the same cost. However, our algorithm can be adapted to work with asymmetric cost functions as well by splitting the edge into two edges of opposite orientation with different costs and only allowing non-negative flows. During Steiner tree computations, merge the two edges and give the merged edge a weight equal to the cost of traversing the edge forward and backward once.
      
      
      Although we provide an optimal algorithm, the constant factor hidden in the time complexity is extremely large, making it impractical for real-world applications. 
      If the fixed topology has cycle rank $r$ and $k$ branch vertices, then the hidden factor is $2^{O(2^r)} + (2r + 1)^r \cdot 2^k$. The bottleneck of $2^{O(2^r)}$ arises from solving the min-cost flow problem on a graph that is a tree with $r$ additional edges. In practice, one can use theoretically slower, off-the-shelf min-cost flow algorithms to avoid the exponential dependency on $r$. However, there are two potential avenues for improvement:
      
      \begin{enumerate}
        \item The problem is very close to linear programming in constant dimensions, so techniques other than Megiddo's algorithm might be beneficial for improving the running time. 
        \item In the proof of \autoref{thm:mincostcycle}, the matrix $A$ is a \emph{network matrix}, but Zemel's algorithm is designed for arbitrary matrices. This presents another opportunity for optimization.
      \end{enumerate}
      The factor $(2r + 1)^r$ is unlikely to be significantly improved if we need to enumerate all elementary circulations, as their number can be exponential with respect to $r$. However, there may be strategies to consider only a much smaller subset of elementary circulations. For example, our enumeration of elementary circulations is independent of cost information, which could potentially be leveraged to reduce the search space.

  \paragraph*{Acknowledgements} We would like to thank Siyue Liu and Jianbo Wang for the initial discussions related to this problem. Chao would like to thank Zexuan Liu and Luze Xu for their discussions on the proximity results of a related integer program.

\bibliographystyle{plain}
\bibliography{graph_ride}
\end{document}